\documentclass[a4paper,twocolumn,english,showpacs,nofootinbib,nobibnotes,aps,pra,superscriptaddress]{revtex4-1}
\usepackage[utf8]{inputenc}
\usepackage[T1]{fontenc}
\usepackage{amsmath}
\usepackage{amsfonts}
\usepackage{amssymb}
\usepackage{amsthm}
\usepackage{braket}
\usepackage{dsfont}

\usepackage{graphicx}

\newtheorem{definition}{Definition}
\newtheorem{lemma}[definition]{Lemma}
\newtheorem{theorem}[definition]{Theorem}
\newtheorem{corollary}[definition]{Corollary}
\newtheorem{observation}[definition]{Observation}

\global\long\def\trace{\operatorname{Tr}}

\begin{document}

\title{Entanglement detection with scrambled data}

\author{Timo Simnacher}

\affiliation{Naturwissenschaftlich-Technische Fakult{\"a}t, 
Universit{\"a}t Siegen, Walter-Flex-Stra{\ss}e~3, 57068 Siegen, 
Germany}

\author{Nikolai Wyderka}

\affiliation{Naturwissenschaftlich-Technische Fakult{\"a}t, 
Universit{\"a}t Siegen, Walter-Flex-Stra{\ss}e~3, 57068 Siegen, 
Germany}

\author{Ren\'e Schwonnek}

\affiliation{Department of Electrical and Computer Engineering, 
National University of Singapore, 4 Engineering Drive 3, Singapore 117576,
Singapore}

\affiliation{Institut f{\"u}r Theoretische Physik, 
Leibniz Universit{\"a}t Hannover, Appelstra{\ss}e~2, 30167 Hannover,
Germany}

\author{Otfried G{\"u}hne}

\affiliation{Naturwissenschaftlich-Technische Fakult{\"a}t, 
Universit{\"a}t Siegen, Walter-Flex-Stra{\ss}e~3, 57068 Siegen, 
Germany}

\date{\today}

\begin{abstract}
In the usual entanglement detection scenario the possible measurements and 
the corresponding data are assumed to be fully characterized. We consider 
the situation where the measurements are known, but the data is scrambled, 
meaning the assignment of the probabilities to the measurement outcomes is 
unknown. We investigate in detail the two-qubit scenario with local 
measurements in two mutually unbiased bases. First, we discuss the use 
of entropies to detect entanglement from scrambled data, showing that
Tsallis- and R\'enyi entropies can detect entanglement in our scenario, 
while the Shannon entropy cannot. Then, we introduce and discuss 
scrambling-invariant families of entanglement witnesses. Finally, we 
show that the set of non-detectable states in our scenario is non-convex 
and therefore in general hard to characterize.
\end{abstract}

\maketitle

\section{Introduction}

The characterization of entanglement is a central problem in many 
experiments. From a theoretical point of view, methods like quantum 
state tomography or entanglement witnesses are available. In practice, 
however, the situation is not so simple, as experimental procedures
are always imperfect, and the imperfections are difficult to characterize.
To give an example, the usual schemes for quantum tomography require
the performance of measurements in a well-characterized basis such as
the Pauli basis, but in practice the measurements may be misaligned in
an uncontrolled manner. Thus, the question arises how to characterize
states with relaxed assumptions on the measurements or on the obtained
data. 

For the case that the measurements are not completely characterized, 
several methods exist to learn properties of quantum states in a 
calibration-robust or even device-independent manner \cite{seevinck, 
squashing, gittsovich, bancal, bancalexp}. But even if the measurements
are well characterized and trustworthy, there may be problems with
the interpretation of the observed probabilities. For instance, in some
ion trap experiments \cite{rowe2001experimental} the individual ions
cannot be resolved, so that some of the observed frequencies cannot 
be uniquely assigned to the measurement operators in a quantum mechanical 
description. More generally, one can consider the situation where the
connection between the outcomes of a measurement and the observed 
frequencies is lost, in the sense that the frequencies are permuted 
in an uncontrolled way. We call this situation the ``scrambled data'' 
scenario. Still it can be assumed that the measurements have a 
well-characterized quantum mechanical description; so the considered
situation is complementary to the calibration-robust or 
device-independent scenario.

In this paper, we present a detailed study of different methods of 
entanglement detection using scrambled data. After explaining the 
setup and the main definitions, our focus lies on the two-qubit case
and Pauli measurements. We first study the use of entropies for 
entanglement detection. Entropies are natural candidates for this
task, as they are invariant under permutations of the probabilities. 
We demonstrate that Tsallis- and R\'enyi entropies can detect 
entanglement in our scenario, while the Shannon entropy is sometimes
useless. For deriving our criteria, we prove some entropic uncertainty
relations, which may be of independent interest.

Second, we introduce scrambling invariant entanglement witnesses. The
key observation is here that for certain witnesses the permutation of
the data corresponds to the evaluation of another witness, so that the
scrambling of the data does not matter.

Third, we characterize the states for which the scrambled data may origin
from a separable state, meaning that their entanglement cannot be 
detected in the scrambled data scenario. We show that this set of
states is generally not convex, which gives an intuition why entanglement 
detection with scrambled data is a hard problem in general.


\section{Setup and Definitions}


Consider an experiment with two qubits and local dichotomic measurements $A \otimes B$, 
the data then consists of four outcome probabilities $p(A=\pm 1, B=\pm 1)$. We define 
the \emph{scrambled data} as a random permutation of these probabilities within but 
not in-between measurements, such that the assignment of probabilities to outcomes 
is forgotten. The restriction that permutations in-between measurements are 
excluded is natural since they are generically inconsistent because the 
probabilities within a measurement do not sum up to one anymore. 

We denote the Pauli matrices by $\sigma_x$, $\sigma_y$, and $\sigma_z$ and 
the corresponding eigenvectors as $\ket{\pm}$, $\ket{\text{y}^\pm}$, and 
$\ket{0},\ket{1}$, respectively. As an example for scrambled data, we 
consider the singlet state $\ket{\psi^-} = (\ket{+-}-\ket{-+})/\sqrt{2}$ 
and the product state $\ket{0} \otimes \ket{+}$. In order to detect the 
entanglement of $\ket{\psi^-}$, it makes sense to perform the local 
measurements $\sigma_x \otimes \sigma_x$ and $\sigma_z \otimes \sigma_z$, 
as there exists an entanglement witness 
$W = \mathds{1} + \sigma_x \otimes \sigma_x + \sigma_z \otimes \sigma_z$ 
detecting this state \cite{guhne2009entanglement}. These measurements yield 
the outcome probabilities $p_{++}$, $p_{+-}$, $p_{-+}$, and $p_{--}$ and 
$p_{00}$, $p_{01}$, $p_{10}$, and $p_{11}$, respectively.

\begin{table}[t]
\renewcommand{\arraystretch}{1.3}
\begin{tabular}{|l|cccc|cccc|}
\hline
 & \multicolumn{4}{ c| }{$\ket{\psi^-} = (\ket{+-}-\ket{-+})/\sqrt{2}$} & \multicolumn{4}{ c| }{$\ket{+} \otimes \ket{0}$}\\
 & $p_{++}$ & $p_{+-}$ & $p_{-+}$ & $p_{--}$ & $p_{++}$ & $p_{+-}$ & $p_{-+}$ & $p_{--}$ \\ \hline
$\sigma_x \otimes \sigma_x$ & $0$ & $\frac{1}{2}$ & $\frac{1}{2}$ & $0$ & $\frac{1}{2}$ & $\frac{1}{2}$ & $0$ & $0$\\
$\sigma_z \otimes \sigma_z$ & $0$ & $\frac{1}{2}$ & $\frac{1}{2}$ & $0$ & $\frac{1}{2}$ & $0$ & $\frac{1}{2}$ & $0$
\\ \hline
\end{tabular}
\caption[justification=justified]{This table shows the measurement data 
for the singlet state and the product state $\ket{+} \otimes \ket{0}$ 
and local measurements $\sigma_x \otimes \sigma_x$ and $\sigma_z \otimes \sigma_z$. 
The scrambled data is the same for the two states. Thus, detecting the entanglement 
of the singlet state with these measurements is impossible using only the scrambled data.}
\label{table1}
\end{table}

From Table~\ref{table1}, we clearly see that the measurement data is different 
for the two states. However, it is easy to see that there is no way of distinguishing 
the two states using these measurements if one has access only to the \textit{scrambled} 
data since the probability distributions are mere permutations of each other. Thus, 
it is impossible to detect the entanglement of the singlet state because there 
is a separable state realizing the same scrambled data. We call states whose 
scrambled data can be realized by a separable state \textit{possibly separable} 
as the entanglement cannot be detected in this scenario.

The above observation motivates to focus specifically on the local 
measurements $\sigma_x \otimes \sigma_x$ and $\sigma_z \otimes \sigma_z$ 
in the following analysis. However, all results hold more generally for 
local measurements $A_1 \otimes B_1$ and $A_2 \otimes B_2$ if the 
eigenstates of both $A_1, A_2$ and $B_1, B_2$ form mutually unbiased bases. 
This is clear from the Bloch sphere representation because any orthogonal 
basis can be rotated to match the analysis in this work. Indeed, in dimension 
two and three, all pairs of mutually unbiased bases are equivalent under 
local unitaries \cite{brierley2009all}, including the locally two-dimensional 
case considered here.


\section{Entropic uncertainty relations}


\begin{figure}[t]
{\includegraphics[width=0.9\linewidth]{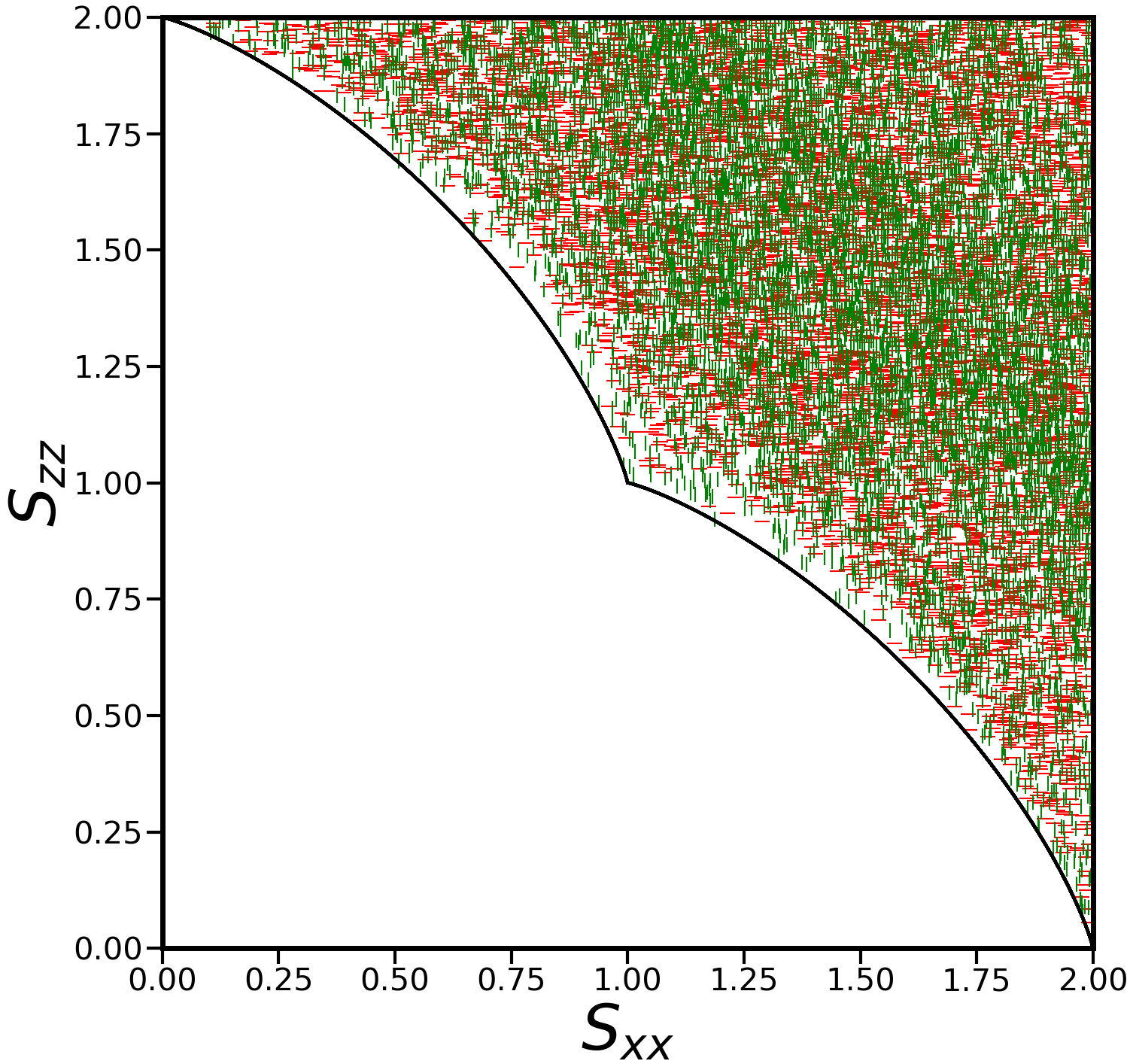}}\\
{\includegraphics[width=0.9\linewidth]{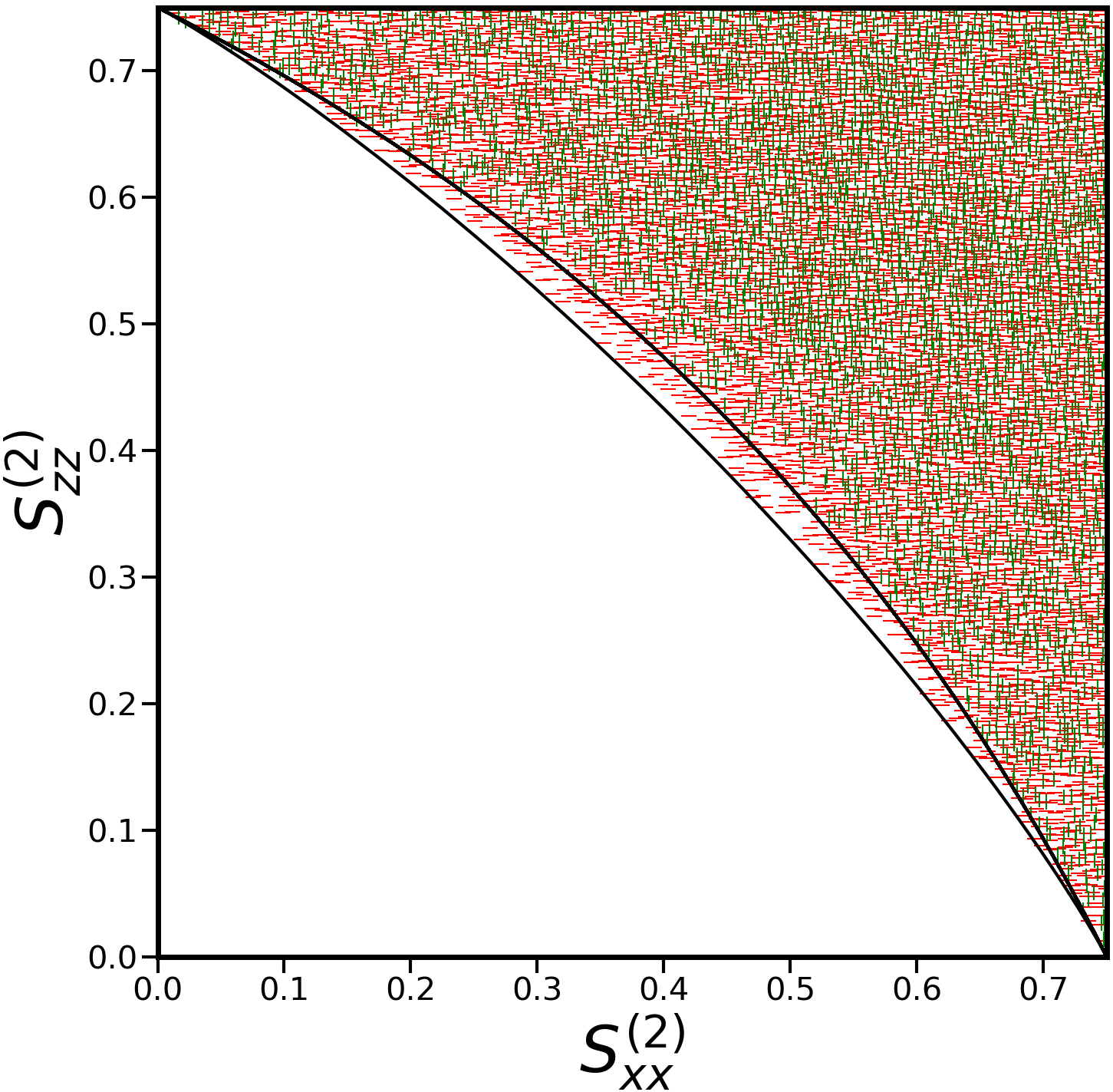}}
\caption[justification=raggedright]{These plots show entropy samples of local 
measurements $\sigma_x \otimes \sigma_x$ and $\sigma_z \otimes \sigma_z$ for 
Shannon entropy (top) and Tsallis-$2$ entropy (bottom) where separable and 
entangled states are represented by green vertical and red horizontal lines, 
respectively. The plot indicates that Shannon entropy is useless for 
entanglement detection, while Tsallis-$2$ entropy is suitable.}
\label{fig:shannontsallis}
\end{figure}

Entropies provide a natural framework to examine scrambled data because 
they are invariant under permutation of probabilities and hence, robust 
against scrambling. In this section, we show that measuring Tsallis-$q$ 
or R\'enyi-$\alpha$ entropies for the two local measurements  
$\sigma_x \otimes \sigma_x$ and $\sigma_z \otimes \sigma_z$ 
in many cases allows for the detection of entanglement and 
show a new family of non-linear, optimal entropic uncertainty relations.

For local measurements $\sigma_i \otimes \sigma_i$, $S_{ii}$ and 
$S_{ii}^{(q)}$ where $i \in \{x,y,z\}$ shall denote the Shannon and 
Tsallis-$q$ entropy of the corresponding four probabilities, 
respectively. For probabilities $\vec{p}$, they are given by 
\cite{shannon1949, havrda1967, tsallis1988possible, wehner2010entropic}
\begin{align}
S(\vec{p}) &= - \sum_j p_j \log p_j, \\
S^{(q)}(\vec{p}) &= \frac{1}{q-1} \Big( 1 - \sum_j p_j^q \Big).
\end{align}

In order to detect entanglement, we investigate the possible pairs 
of $S_{xx}^{(\tilde{q})}$ and $S_{zz}^{(q)}$ that can be realized 
by physical states. For gaining some intuition, we have plotted in 
Fig.~\ref{fig:shannontsallis} random samples of separable and entangled 
two-qubit states, where separability can be checked using the PPT 
criterion \cite{peres1996separability}. As the figures indicate, 
the accessible region for both kinds of states does not differ in 
the case of Shannon entropy and hence, entanglement detection seems 
impossible in this case. This is supported by findings in earlier works:
It has been shown in Ref.~\cite{schwonnek2018additivity} that in the 
case of Shannon entropy and two local measurements, linear entropic 
uncertainty relations of the type $\alpha S_{xx} + \beta S_{zz} \geq c_{\text{sep}} \geq c$ 
with bounds $c_{\text{sep}}$ for separable and $c$ for all states, 
are infeasible to detect entanglement, i.e., $c_{\text{sep}} = c$. 
Furthermore, Conjecture V.6 in Ref.~\cite{abdelkhalek2015optimality} 
states that in the example of local measurements $\sigma_x \otimes \sigma_x$ 
and $\sigma_z \otimes \sigma_z$, even non-linear entropic uncertainty 
relations cannot be used to detect entanglement. However, non-linear 
relations are unknown in most cases \cite{abdelkhalek2015optimality}.

In contrast to the case of Shannon entropy, using Tsallis-$2$ entropy, 
we identify a distinct region occupied by entangled states only, see the
lower part of Fig.~\ref{fig:shannontsallis}). In the following, we will 
show that also $(S_{xx}^{(\tilde{q})},S_{zz}^{(q)})$-plots with 
$q,\tilde{q}\geqslant 2$ exhibit this feature by determining the lower 
bounds of the set of all and the set of separable states.

In order to obtain the boundary of this realizable region, note first 
that a vanishing entropy of $S_{ii}^q = 0$ implies that the system 
is in an eigenstate of the measurement operator $\sigma_i \otimes \sigma_i$. 
Since the measurements define mutually unbiased bases, it is clear 
that in this case the other entropy is maximal. Hence, the states 
$\ket{00}$ and $\ket{++}$ lie on the boundary of the realizable 
region. The mixture of these states with white noise $\mathds{1}/4$ 
leaves the maximal entropy of one measurement unchanged while 
increasing the entropy of the other measurement continuously. 
Therefore, the upper and the right boundary of the region, 
corresponding to maximal $S_{zz}$ and $S_{xx}$, respectively, 
is reached by separable states (see Fig.~\ref{fig:shannontsallis}). 
We will see later that the lower boundaries for all and for separable 
states are both realized by continuous one-parameter families of states. 
Thus, the mixture of these states with white noise forms a continuous 
family of curves connecting the lower boundary with the point where 
both entropies are maximal. Hence, these states realize any accessible 
point in the entropy plot and it is sufficient to only determine 
the lower boundary.

\subsection{Entropic bound for general states}

We begin by determining the bounds in the $(S_{xx}^{(\tilde{q})},S_{zz}^{(q)})$-plot 
for all states.

In Ref.~\cite{abdelkhalek2015optimality}, Theorem V.2 states that for two concave functionals $f_1,f_2$ on the state space, for any state $\rho$, there is a pure state $\ket{\psi}$ such that $f_1(\ket{\psi}\bra{\psi}) \leqslant f_1(\rho)$ and $f_2(\ket{\psi}\bra{\psi}) \leqslant f_2(\rho)$. Furthermore, it is shown in Theorem V.3 that the state $\ket{\psi}$ can additionally be chosen real if the inputs of the functionals are linked by a real unitary matrix.
Thus, in the case of general two-qubit states and local measurements $\sigma_x \otimes \sigma_x$ and $\sigma_z \otimes \sigma_z$, the analysis of the boundary of the entropy plots can be reduced to pure real states. 
First, we will solve the special case where $q=2$, also known as linear entropy. This result can then be used as an anchor to prove the bound for all $q \geqslant 2$.

\begin{lemma}\label{linearEnt}
For two-qubit states $\rho$ and fixed $S_{zz}^{(2)}(\rho)$, minimal $S_{xx}^{(2)}(\rho)$ is reached by the unique state $\rho = \ket{\psi_t}\bra{\psi_t}$ where $\ket{\psi_t} = \frac{1}{\sqrt{3+t^2}} (t \ket{00} + \ket{01} + \ket{10} + \ket{11})$ and some $t \geqslant 1$ determined by the given entropy $S_{zz}^{(2)}(\rho)$.
\end{lemma}
\begin{proof}
According to Theorem V.3 in Ref.~\cite{abdelkhalek2015optimality}, if two entropies $S_1$ and $S_2$ are considered where the measurement bases are related by a real unitary transformation, then for any state $\rho$, there is always a pure and real state $\ket{\psi}$ with $S_1(\ket{\psi}\bra{\psi}) \leqslant S_1(\rho)$ and $S_2(\ket{\psi}\bra{\psi}) \leqslant S_2(\rho)$. 
As in our case $\sigma_x = H \sigma_z H^\dagger$ where $H$ is the Hadamard matrix, it is sufficient to consider pure real states to obtain minimal $S_{zz}^{(2)}$ for given $S_{xx}^{(2)}$. For a general pure real state $\ket{\psi} = (x_1,x_2,x_3,x_4)^T$, the problem boils down to the following maximization problem under constraints
\begin{align}\begin{split}
\max_{x_i} \: &f(x_1,x_2,x_3,x_4),\\
\text{s.t. } &x_1^4+x_2^4+x_3^4+x_4^4 = k = \text{const.},\\
&x_1^2+x_2^2+x_3^2+x_4^2 = 1
\end{split}\end{align}
where $f(\{x_i\}) = (x_1+x_2+x_3+x_4)^4 + (x_1+x_2-x_3-x_4)^4 + (x_1-x_2+x_3-x_4)^4 + (x_1-x_2-x_3+x_4)^4 = 1 - S_{xx}^{(2)}$ and $k = 1 - S_{zz}^{(2)}$. Note that $\frac{1}{4} \leqslant k \leqslant 1$. It is straightforward to see that $\frac{1}{96} [f(\{x_i\}) - 12 \times 1^2 + 8 k] = \frac{1}{96} [f(\{x_i\}) - 12 (x_1^2+x_2^2+x_3^2+x_4^2)^2 + 8 (x_1^4+x_2^4+x_3^4+x_4^4)] = x_1 x_2 x_3 x_4$, using the constraints. So, we can replace $f$ by $x_1 x_2 x_3 x_4$. Clearly, the $x_j$ can be chosen greater than $0$ in case of a maximum. Consequently, $x_i$ can be substituted by $\sqrt{y_i}$. Because $\sqrt{\cdot}$ is a monotone function, the objective function  $y_1 y_2 y_3 y_4$ is equivalent to $\sqrt{y_1 y_2 y_3 y_4}$ since we are only interested in the state realizing the boundary and not necessarily the boundary itself. Thus, the problem reduces to
\begin{align}\begin{split}
\max_{y_i} \: &y_1 y_2 y_3 y_4,\\
\text{s.t. } &y_1^2+y_2^2+y_3^2+y_4^2 = k = \text{const.},\\
&y_1+y_2+y_3+y_4 = 1
\end{split}\end{align}
where all $y_i$ are positive. Using Lagrange multipliers, one obtains the optimal solution: For given $S_{zz}^{(2)}$, the minimal $S_{xx}^{(2)}$ is reached by the state
\begin{align}\label{psit}
\ket{\psi_t} = \frac{1}{\sqrt{3+t^2}} (t \ket{00} + \ket{01} + \ket{10} + \ket{11}),
\end{align}
for some $t \geqslant 1$. Since the minimal $S_{zz}^{(2)}$-entropy state $\ket{\psi_\infty} = \ket{00}$ and maximal $S_{zz}^{(2)}$-entropy state $\ket{\psi_0} = \frac{1}{2} (\ket{00} + \ket{01} + \ket{10} + \ket{11})$ are part of the family $\ket{\psi_t}$ and $\frac{dS_{zz}^{(2)}(\ket{\psi_t})}{dt} < 0$, fixing $S_{zz}^{(2)}$ uniquely determines $t$ and hence, also $\rho_t = \ket{\psi_t}\bra{\psi_t}$.
\end{proof}

This result holds for the Tsallis-2 entropy. However, it can be generalized 
to any pair of Tsallis-$q$ and Tsallis-$\tilde{q}$ entropies with 
$q,\tilde{q} \geqslant 2$. To that end, we use a result from Ref.~\cite{berry2003bounds}. 
There, the authors consider entropy measures $H_f = \sum_i f(p_i)$ and 
$H_g = \sum_i g(p_i)$ where $f(0) = g(0) = 0$ and the functions $f,g$ are 
strictly convex (implying that $g'(p)$ is invertible) with their first 
derivatives being continuous in the 
interval $(0,1)$. They show that then the maximum (minimum) of 
$H_f$ for fixed $H_g$ is obtained by the probability distribution 
$p_1 \geqslant p_2 = \dots = p_n$ if $f^\prime[p(g')]$ as a function 
of $g'$ is strictly convex (concave). Furthermore, for each value 
of $H_g$, there is a unique probability distribution of this form.

In the specific case of Tsallis entropies with parameters $q$ and $\tilde{q}$, it is shown that if
\begin{align}
\frac{q(q-1)}{\tilde{q}(\tilde{q}-1)} p^{q-\tilde{q}}
\end{align}
is monotonically increasing (decreasing), the minimum (maximum) $S^{q}$ for fixed $S^{\tilde{q}}$ is reached by the probability distribution described above, when considering the same measurement for different $q, \tilde{q}$. That is exactly the probability distribution obtained by measuring $\sigma_x \otimes \sigma_x$ or $\sigma_z \otimes \sigma_z$ locally in the state $\ket{\psi_t}$ since
\begin{align}
\ket{\psi_t} &\propto t \ket{00} + \ket{01} + \ket{10} + \ket{11} \nonumber \\ 
&\propto (t+3) \ket{++} + (t-1) (\ket{+-} + \ket{-+} + \ket{--}),
\end{align}
where $t^2 \geqslant 1$ and $(t+3)^2 \geqslant (t-1)^2$. 
This observation assists in proving the following theorem:

\begin{figure}[t]
{\includegraphics[width=0.97\linewidth]{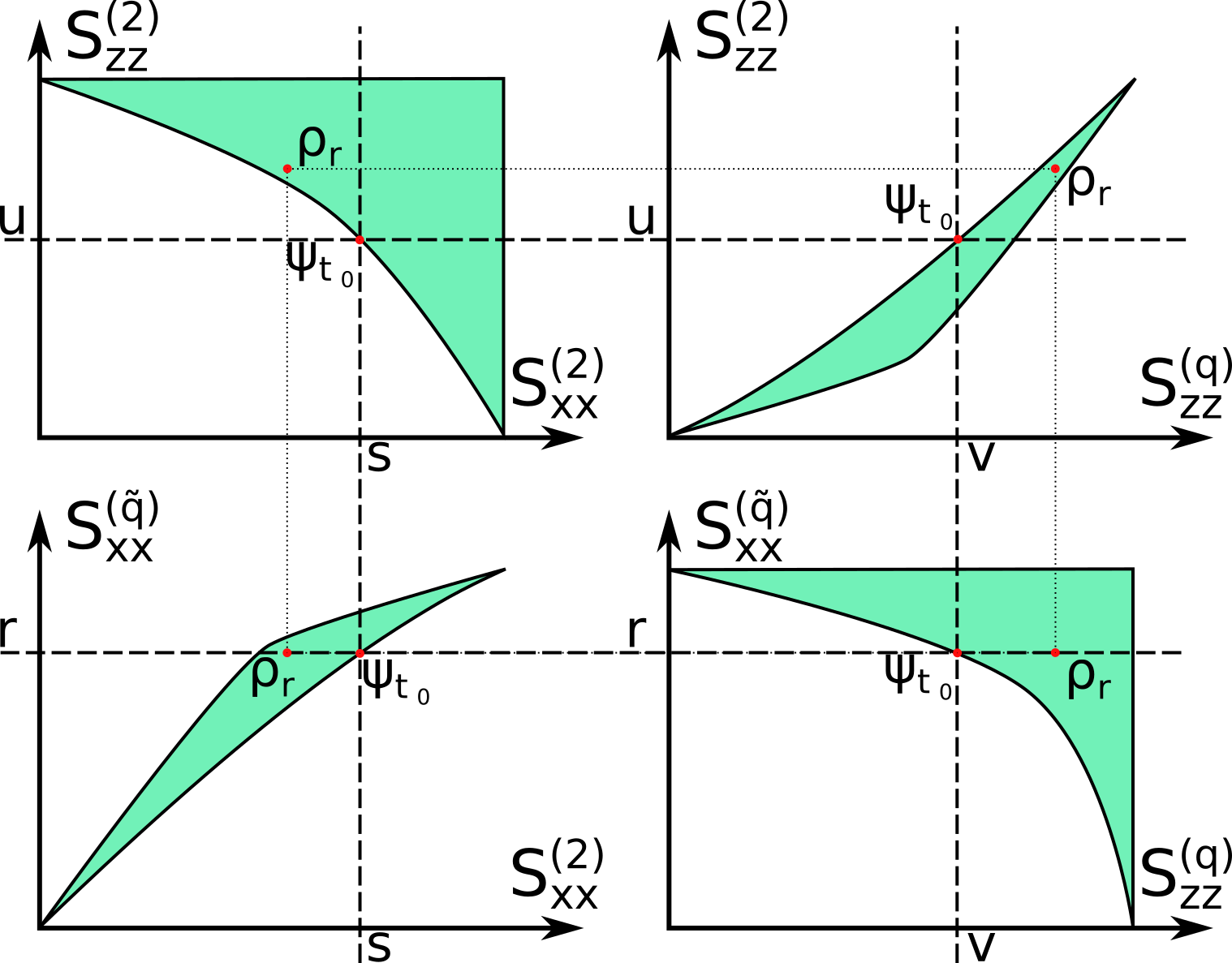}}\\
\caption[justification=raggedright]{These sketched plots 
depict the proof of Theorem~\ref{thm:qtildeq}. Starting with
the lower right picture, for fixed $\rho_r$ with $S_{xx}^{(\tilde{q})}(\rho_r)=r$, 
we consider the state $\ket{\psi_{t_0}}$ defined in Lemma~\ref{linearEnt}, 
with $t_0$ such that also $S_{xx}^{(\tilde{q})}(\psi_{t_0}) = r$. The state 
$\ket{\psi_{t_0}}$ has the largest $S_{xx}^{(2)}$-entropy among all states 
$\rho$ with $S_{xx}^{(\tilde{q})}(\rho) = r$ \cite{berry2003bounds}, 
particularly including $\rho_r$ (see lower left). From Lemma~\ref{linearEnt}, 
it follows that $S_{zz}^{(2)}(\psi_{t_0}) \leqslant S_{zz}^{(2)}(\rho_r)$ 
which is shown in the upper left. This, in turn, implies that 
$S_{zz}^{(q)}(\psi_{t_0}) \leqslant S_{zz}^{(q)}(\rho_r)$ 
\cite{berry2003bounds} (see plot on the upper right).
In summary, we have for any state $\rho_r$ that there exists 
a state $\ket{\psi_{t_0}}$ with 
$S_{xx}^{(\tilde{q})}(\psi_{t_0}) = S_{xx}^{(\tilde{q})}(\rho)$ 
and $S_{zz}^{(q)}(\psi_{t_0}) \leqslant S_{zz}^{(q)}(\rho)$. 
This proves that the boundary is realized by the states 
$\ket{\psi_t}$, which is illustrated again in the lower right.}
\label{fig:berrysanders}
\end{figure}

\begin{theorem}\label{thm:qtildeq}
For all $q,\tilde{q} \geqslant 2$, the lower boundary in the $(S_{xx}^{(\tilde{q})}, S_{zz}^{(q)})$-plot is realized by the family of states $\ket{\psi_t} = \frac{1}{\sqrt{3+t^2}} (t \ket{00} + \ket{01} + \ket{10} + \ket{11})$ where $t \geqslant 1$.
\end{theorem}
\begin{proof}
For fixed $r$ and a state $\rho_r$ with $S_{xx}^{(\tilde{q})}(\rho_r) = r$, there exists a unique state $\ket{\psi_{t_0}}$ with  $S_{xx}^{(\tilde{q})}(\psi_{t_0}) = r$. From Theorem $1$ in Ref.~\cite{berry2003bounds}, it follows that
\begin{align}
S_{xx}^{(2)}(\rho_r) \leqslant S_{xx}^{(2)}(\psi_{t_0}) \equiv s
\end{align}
(see bottom left graph in Fig.~\ref{fig:berrysanders}).
Since $\frac{dS_{zz}^{(2)}(\ket{\psi_t})}{dS_{xx}^{(2)}(\ket{\psi_t})} = \frac{dS_{zz}^{(2)}(\ket{\psi_t})}{dt} \left(\frac{dS_{xx}^{(2)}(\ket{\psi_t})}{dt}\right)^{-1} < 0$, it follows from Lemma~\ref{linearEnt} that $S_{xx}^{(2)}(\rho_r) \leqslant s$ implies 
\begin{align}
S_{zz}^{(2)}(\rho_r) \geqslant S_{zz}^{(2)}(\psi_{t_0}) \equiv u
\end{align}
(see top left graph in Fig.~\ref{fig:berrysanders}). Now, given $S_{zz}^{(2)}(\rho_r) \geqslant u$, using the fact that $\frac{dS_{zz}^{(q)}(\ket{\psi_t})}{dS_{zz}^{(2)}(\ket{\psi_t})} > 0$, it follows from Ref.~\cite{berry2003bounds} that
\begin{align}
S_{zz}^{(q)}(\rho_r) \geqslant S_{zz}^{(q)}(\psi_{t_0}) \equiv v
\end{align}
(see top right graph in Fig.~\ref{fig:berrysanders}).

\begin{figure}[t]
{\includegraphics[width=0.97\linewidth]{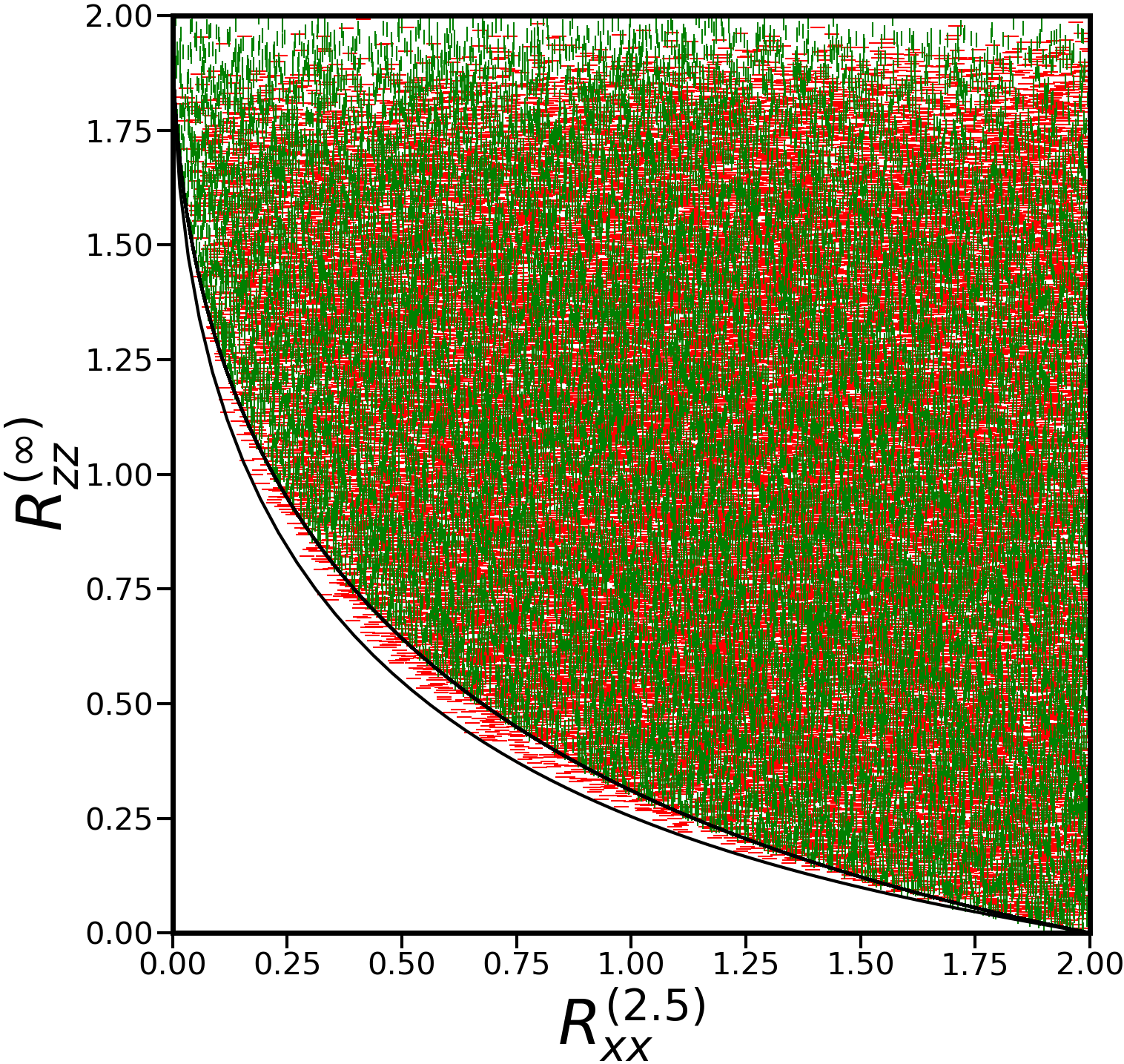}}\\
\caption[justification=raggedright]{This plots shows entropy samples of local measurements $\sigma_x \otimes \sigma_x$ and $\sigma_z \otimes \sigma_z$ for R\'enyi-$2.5$- and R\'enyi-$\infty$-entropies, respectively. Separable states are represented by green vertical lines, while red horizontal lines indicate entangled states. The lower boundary is given by the states $\ket{\psi_t}$ defined in Lemma~\ref{linearEnt}.}
\label{fig:renyi25infbw}
\end{figure}

In summary, by considering all values of $r$, we found that for all two-qubit states $\rho_r$,
\begin{align}
S_{xx}^{(\tilde{q})}(\rho_r) = r &\Rightarrow S_{xx}^{(2)}(\rho_r) \leqslant S_{xx}^{(2)}(\psi_{t_0})\\
&\Rightarrow S_{zz}^{(2)}(\rho_r) \geqslant S_{zz}^{(2)}(\psi_{t_0})\\
&\Rightarrow S_{zz}^{(q)}(\rho_r) \leqslant S_{zz}^{(q)}(\psi_{t_0})
\end{align}
(see also the lower right graph in Fig.~\ref{fig:berrysanders}), where $\ket{\psi_{t_0}}$ is uniquely determined by $S_{xx}^{(\tilde{q})}(\psi_{t_0}) = r$. All bounds, as well as the overall implication $S_{xx}^{(\tilde{q})}(\rho_r) = r \Rightarrow S_{zz}^{(q)}(\rho_r) \leqslant S_{zz}^{(q)}(\psi_{t_0})$ are tight since they are saturated by the same state $\ket{\psi_{t_0}}$. Thus, the lower boundary in the $(S_{xx}^{(\tilde{q})}, S_{zz}^{(q)})$-plot is realized by the family of states $\ket{\psi_t} = \frac{1}{\sqrt{3+t^2}} (t \ket{00} + \ket{01} + \ket{10} + \ket{11})$ where $t \geqslant 1$.
\end{proof}

In the above proof, we used the $q=\tilde{q}=2$-case as an anchor to derive the result for all $q,\tilde{q} \geqslant 2$. The same argument also holds if we would use any other anchor case where the $\ket{\psi_t}$ are the optimal states. Numerical evidence suggests that the conclusion is indeed valid for any $q,\tilde{q} \gtrsim 1.37$.
Furthermore, the result can also be interpreted as a family of entropic uncertainty relations.

\begin{corollary}\label{cor:eur}
For all two-qubit states $\rho$ and $q,\tilde{q}\geqslant2$, 
\begin{align}
&F[S_{xx}^{(\tilde{q})}(\rho),S_{zz}^{(q)}(\rho)] \equiv S_{zz}^{(q)}(\rho) - S_{zz}^{(q)}(\ket{\psi_t}[S_{xx}^{(\tilde{q})}(\rho)]) 
\nonumber \\
&= S_{zz}^{(q)}(\rho) - \frac{1}{q-1}\left(1-\frac{3+t^{2q}[S_{xx}^{(\tilde{q})}(\rho)]}{\{3+t^2[S_{xx}^{(\tilde{q})}(\rho)]\}^q}\right) \geqslant 0. 
\label{eq:eur}
\end{align} 
Here, $\ket{\psi_t}[S_{xx}^{(\tilde{q})}(\rho)]$ and 
$t[S_{xx}^{(\tilde{q})}(\rho)]$ are the unique state 
$\ket{\psi_t}$ and parameter $t$ in dependence on
$S_{xx}^{(\tilde{q})}(\rho)$, such that 
$S_{xx}^{(\tilde{q})}(\ket{\psi_t}) = S_{xx}^{(\tilde{q})}(\rho)$.
\end{corollary}

For the example of $q = \tilde{q} = 2$, we have
\begin{align}
F[S_{xx}^{(2)}(\rho),S_{zz}^{(2)}(\rho)]=S_{zz}^{(2)}-\frac{3QT^{2}-T^{4}}{3Q^{2}}\geqslant0,
\end{align}
where
\begin{align}
T &= \sqrt{9-12S_{xx}^{(2)}},\\
Q &= 3+T+\sqrt{3}\sqrt{(1+T)(3-T)}.
\end{align}
This bound is displayed in Fig.~\ref{fig:shannontsallis}.

Note that this result is also valid for R\'enyi-$\alpha$ entropies 
\cite{renyi1961proceedings} with $\alpha, \tilde{\alpha} \geqslant 2$ as R\'enyi-$\alpha$ and Tsallis-$q$ entropies are monotone functions of each other for $\alpha = q$. Thus, the change from the Tsallis- to the R\'enyi entropy induces simply a rescaling of the axis in the $(S_{xx},S_{zz})$-plot. An example is given in Fig.~\ref{fig:renyi25infbw} where $\alpha = 2.5$ and $\tilde{\alpha} = \infty$.

In contrast to any linear bounds, which are usually considered \cite{schwonnek2018additivity}, the uncertainty relations found here are optimal. That means, for any entropic uncertainty relation defined in Corollary~\ref{cor:eur} and any $S_{zz}^{(q)}$, there exists a state, namely the $\ket{\psi_t}$ with the given entropy, saturating the corresponding bound.

\subsection{Entropic bound for separable states}

In this section, we determine the bound for separable states. 
Theorem V.2 from Ref.~\cite{abdelkhalek2015optimality}, which shows that for any state $\rho$ there is a pure state $\ket{\psi}$ such that $f_1(\psi) \leqslant f_1(\rho)$ and $f_2(\psi) \leqslant f_2(\rho)$, cannot be applied to separable states. 
This is because the boundary of the space of separable states is determined by positivity as well as separability conditions. 
While the former implies that states on the boundary are of lower rank, the latter gives a different constraint. 
However, this can still be used to simplify the optimization process, as we prove in the following:

\begin{theorem}\label{thm:proofsep}
Let $f_1$, $f_2$ be two continuous concave functions on the state space. Then, for every separable state $\rho$, there exists a separable state $\rho^*$ of the form
\begin{align}
\rho^* = (1-p) \ket{ab}\bra{ab} + p \ket{cd}\bra{cd},
\end{align}
where $0 \leqslant p \leqslant 1$ and $\ket{ab}\bra{ab}$, $\ket{cd}\bra{cd}$ are pure product states, such that $f_1(\rho^*) \leqslant f_1(\rho)$ and $f_2(\rho^*) \leqslant f_2(\rho)$.
\end{theorem}

\begin{proof}
In the range of $\rho$, we consider some state $\sigma$ on the boundary of the space of separable states in this subspace. 
Then, there is some antipode $\sigma^\blacktriangledown$ defined as $\frac{1}{\lambda} \left[ \rho - (1-\lambda)\sigma \right]$ for the smallest $\lambda$ such that this expression still describes a separable state. 
By this definition, obviously, also $\sigma^\blacktriangledown$ lies on the boundary. 
Now, $\sigma$ can be converted continuously into $\sigma^\blacktriangledown$ by a curve $t \mapsto \gamma(t)$ on the  boundary where $\gamma(0) = \sigma$ and $\gamma(1) = \sigma^\blacktriangledown$, as long as the boundary is connected (see Fig.\ref{fig:proofsep}). 
Since the functions are continuous, there must be some $t^* \in [0,1]$ such that $f_1[\gamma(t^*)] = f_1(\rho)$. 
At this point, either $f_2[\gamma(t^*)] \leqslant f_2(\rho)$ or it holds that $f_1[\gamma^\blacktriangledown (t^*)] \leqslant f_1(\rho)$ and $f_2[\gamma^\blacktriangledown (t^*)] \leqslant f_2(\rho)$ since otherwise concavity implies the contradiction
\begin{align}
f_i(\rho) &\geqslant [1-\lambda(t)] f_i[\gamma(t)] + \lambda(t) f_i[\gamma^\blacktriangledown (t)]\\
 &> [1-\lambda(t)] f_i(\rho) + \lambda(t) f_i(\rho) = f_i(\rho)
\end{align}
for $i=1$ or $i=2$. 
Thus, we find a state $\gamma^*$ with $f_{1,2}[\gamma^*] \leqslant f_{1,2}(\rho)$.
Compared to $\rho$, this boundary state $\gamma^*$ satisfies at least one additional constraint of the form
\begin{align}\label{eq:sep_bound_constr}
\gamma^* \ket{\phi_0} = 0, &&
\trace(\gamma^* W) = 0,
\end{align}
where $\ket{\phi_0}$ is an eigenstate of $\gamma^*$ and $W$ is an entanglement witness, because $\gamma^*$ lies at the positivity or separability boundary, respectively.

When we decompose $\gamma^*$ into pure product states $\gamma^* = \sum_j p_j \ket{a_j b_j}\bra{a_j b_j}$, every $\ket{a_j b_j}\bra{a_j b_j}$ satisfies the constraints individually.
This is because the range of each of them has to be contained in the range of $\gamma^*$. 
Furthermore, for product states it holds that $\trace (\ket{a_j b_j}\bra{a_j b_j} W )\geqslant 0$ and since we have $0 = \trace(\gamma^* W) = \sum_j p_j \trace (\ket{a_j b_j}\bra{a_j b_j} W)$, also $\trace ( \ket{a_j b_j}\bra{a_j b_j} W) = 0$.

Thus, we can apply this procedure repeatedly, considering only the state space defined by the already accumulated constraints of the form given in Eq.~(\ref{eq:sep_bound_constr}). 
In the end, we either have a pure product state $\rho^*$ or a one-dimensional state space spanned by two pure product states $\ket{ab}\bra{ab}$ and $\ket{cd}\bra{cd}$, whose boundary is disconnected and hence, the scheme cannot be applied anymore. 
This might indeed happen, as there are two-dimensional subspaces of the two-qubit 
space with exactly two product vectors in it \cite{sanpera1998local}. 
Either way, for any separable state $\rho$ we find a state $\rho^*$ of the form $\rho^* = (1-p) \ket{ab}\bra{ab} + p \ket{cd}\bra{cd}$ such that $f_1(\rho^*) \leqslant f_1(\rho)$ and $f_2(\rho^*) \leqslant f_2(\rho)$.
\end{proof}

\begin{figure}[t]
{\includegraphics[width=0.97\linewidth]{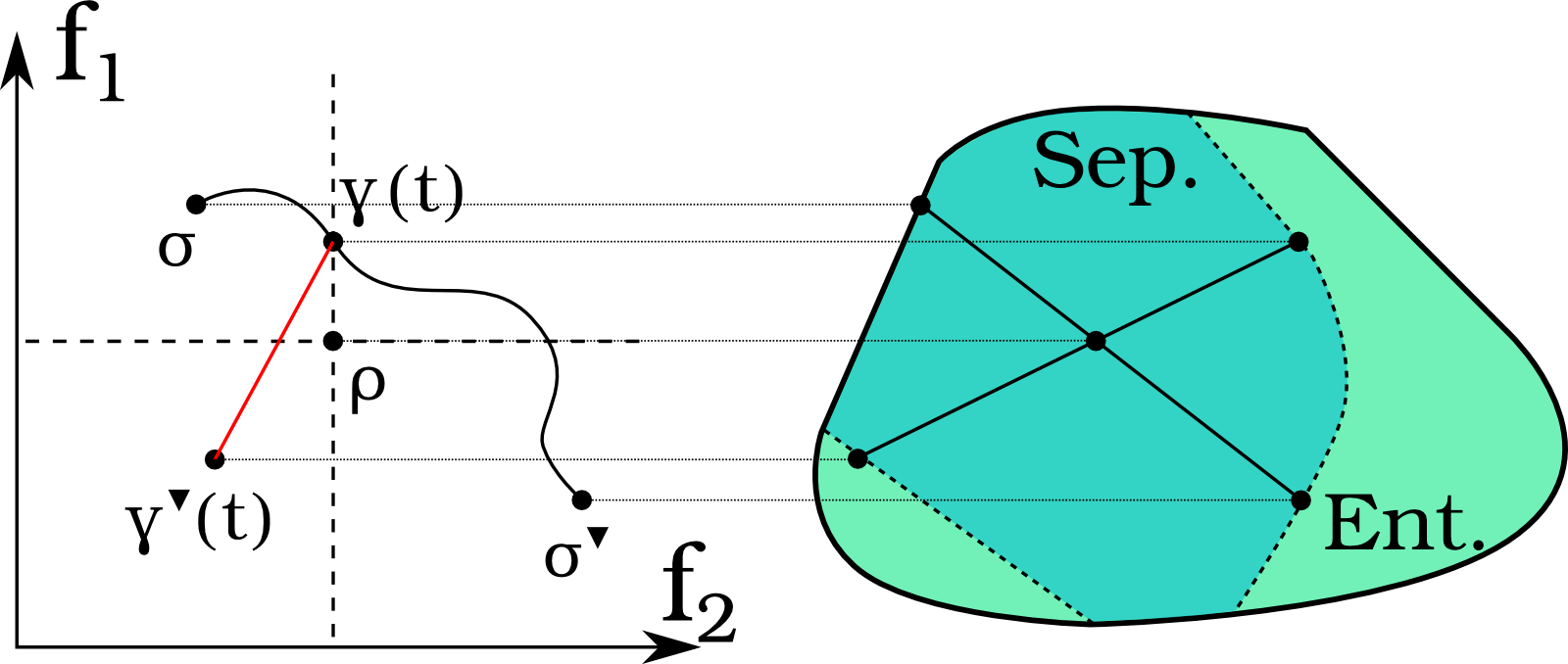}}\\
\caption[justification=raggedright]{This sketch shows the proof idea of Theorem~\ref{thm:proofsep}, where the left plot is based on Fig.~6 in Ref.~\cite{abdelkhalek2015optimality}. Any separable state $\rho$ can be written as the mixture of two states on the topological boundary of the space of separable states. These two states can be converted into each other continuously. In this process, we find a state $\gamma^\blacktriangledown 
(t)$ on the boundary such that $f_1[\gamma^\blacktriangledown (t)] \leqslant f_1(\rho)$ and $f_2[\gamma^\blacktriangledown (t)] \leqslant f_2(\rho)$ for continuous concave functionals $f_1$ and $f_2$.}
\label{fig:proofsep}
\end{figure}

In the case of local $\sigma_x \otimes \sigma_x$ and $\sigma_z \otimes \sigma_z$ measurements, we can restrict the optimization further to real states $\ket{ab}$ and $\ket{cd}$.

\begin{observation}
For any separable state $\rho$, there is a state $\rho^* = (1-p) \ket{ab}\bra{ab} + p \ket{cd}\bra{cd}$ where $0 \leqslant p \leqslant 1$ and $\ket{ab}$ and $\ket{cd}$ are pure and real product states such that $S_{xx}^{(\tilde{q})}(\rho^*) \leqslant S_{xx}^{(\tilde{q})}(\rho)$ and $S_{zz}^{(q)}(\rho^*) \leqslant S_{zz}^{(q)}(\rho)$ for any $q, \tilde{q} \in \mathds{R}$.
\end{observation}

\begin{proof}
Using Theorem~\ref{thm:proofsep}, we immediately find a state $\sigma = (1-p) \ket{ab}\bra{ab} + p \ket{cd}\bra{cd}$ such that $S_{xx}^{(\tilde{q})}(\sigma) \leqslant S_{xx}^{(\tilde{q})}(\rho)$ and $S_{zz}^{(q)}(\sigma) \leqslant S_{zz}^{(q)}(\rho)$ for any $q, \tilde{q} \in \mathds{R}$. 
However, the states $\ket{ab}$ and $\ket{cd}$ might not be real. 

A general one-qubit state can be written as $\ket{a} = \cos\frac{\theta}{2} \ket{0} + e^{i\varphi} \sin\frac{\theta}{2} \ket{1}$ where $0 \leqslant \theta \leqslant \pi$ and $0 \leqslant \varphi < 2\pi$. 
The corresponding probabilities for $\sigma_x$ and $\sigma_z$ measurements are then given by
\begin{align}
p_0 &= \cos^2 \frac{\theta}{2} \text{, } p_1 = \sin^2 \frac{\theta}{2}, \\
p_\pm &= \frac12 \pm \frac12 \sin\theta \cos\varphi.
\end{align}
Hence, by just varying $\varphi$, $p_0$ and $p_1$ remain unaffected, while $p_+ = (1-\alpha) P_{\text{max}} + \alpha P_{\text{min}}$ and $p_- = \alpha P_{\text{max}} + (1-\alpha) P_{\text{min}}$, where $P_{\text{max}} = \frac12 + \frac12 \sin\theta$ and $P_{\text{max}} = \frac12 - \frac12 \sin\theta$, vary continuously with $0 \leqslant \alpha \leqslant 1$. 
Now, consider varying the state $\ket{a}$ in such a way while leaving $\ket{b}$ and $\ket{cd}$ the same.  
Obviously, the probability distribution for the $\sigma_z \otimes \sigma_z$ measurement on $\sigma$ stays unchanged. 
The $\sigma_x \otimes \sigma_x$ measurement, on the other hand, yields $(1-\alpha) \vec{p_1} + \alpha \vec{p_2}$ for some probability distributions $\vec{p_1}$ and $\vec{p_2}$. 
Hence, the optimization problem over $\alpha$ is an optimization over a convex set of probabilities. As the entropies are concave functions of probability distributions, the optimum can be found at the boundary. Note that we only optimize the $S_{xx}^{\tilde{q}}$ while leaving $S_{zz}^{q}$ unchanged. Thus, $\ket{a}$ can be chosen real and so can $\ket{b}$, $\ket{c}$ and $\ket{d}$.
\end{proof}

Reducing the optimization to real states of rank at most two, the lower number of parameters allows for robust numerical analysis. This suggests that for $q,\tilde{q} \geqslant 2$, the boundary is reached by real pure product states of the form
\begin{align}
\ket{\phi^{q,\tilde{q}}_\theta} = \Big( \cos\frac{\theta}{2} \ket{0} + \sin\frac{\theta}{2} \ket{1} \Big)^{\otimes 2}.
\end{align}
For $q = \tilde{q} = 2$, we obtain the (numerical) boundary for separable states $\rho$ as
\begin{align}
S_{zz}^{(2)}(\rho) \geqslant -\frac{9}{4} + 3 \sqrt{1-S_{xx}^{(2)}(\rho)} + S_{xx}^{(2)}(\rho),
\end{align}
which is shown in Fig.~\ref{fig:shannontsallis}. In the case of Shannon entropy, numerical analysis indicates that the boundary is realized by the states
\begin{align}
\ket{\phi^{S}_\theta} &= \ket{0} \otimes \Big( \cos\frac{\theta}{2} \ket{0} + \sin\frac{\theta}{2} \ket{1} \Big), \\
\ket{\psi^{S}_\theta} &= \Big( \cos\frac{\theta}{2} \ket{0} + \sin\frac{\theta}{2} \ket{1} \Big) \otimes \ket{+},
\end{align}
which is also shown in Fig.~\ref{fig:shannontsallis}.

\subsection{Robustness}

In the previous sections, we showed that the accessible regions in the entropy plot $(S^{(q)}_{xx},S^{(\tilde{q})}_{zz})$ are different for general two-qubit states and separable states when $q,\tilde{q} \geqslant 2$. Thus, these entropies provide a scrambling-invariant method to detect entanglement. The accessible regions for $q = \tilde{q} = 2$ are shown in Fig.~\ref{fig:shannontsallis}.

We investigate the robustness of this detection method for different $q=\tilde{q} \geqslant 2$. The robustness is quantified by the amount of white noise that can be added to the boundary states defined in Eq.~(\ref{psit}) such that they are still detectable. Numerical analysis indicates that independent of $q$, the most robust states are those with $S^{(q)}_{xx} = S^{(q)}_{zz}$, i.e. $t = 3$. For states $\rho_{\lambda,t} = (1-\lambda) \ket{\psi_t}\bra{\psi_t} + \lambda \frac{\mathds{1}}{4}$, it also holds that $S^{(q)}_{xx}(\rho_{\lambda,t}) = S^{(q)}_{zz}(\rho_{\lambda,t})$ independent of $\lambda$ and hence, they enter the region of separable states at the point of the symmetric real pure product state $\big[\frac{1}{\sqrt{1+s^2}} (s \ket{0} + \ket{1}) \big]^{\otimes 2}$ where $s = 1 + \sqrt{2}$. The maximal noise level $\lambda$ is then determined by
\begin{align}
\begin{split}
&\Big( \frac{(1-\lambda)t}{\sqrt{3+t^2}} + \frac{\lambda}{4} \Big)^{2q} + 3 \Big( \frac{(1-\lambda)}{\sqrt{3+t^2}} + \frac{\lambda}{4} \Big)^{2q}\\ &= \Big( \frac{s^2}{1+s^2} \Big)^{2q} + 2 \Big( \frac{s}{1+s^2} \Big)^{2q} + \Big( \frac{1}{1+s^2} \Big)^{2q}
\end{split}
\end{align}
which can be solved analytically for large $q$. In the limit of $q \rightarrow \infty$, $\lambda = \frac{1}{11} (10-\sqrt{2}-\sqrt{12}-\sqrt{24}) \approx 0.020$. 
Note that this is an upper bound on the robustness, since the boundary of the region of separable states was only determined numerically in the last section. However, even this upper bound is rather small and the method is not very robust. Finally, we see that the method is most robust for large $q$, but the limit is reached very fast.


\section{Scrambling-invariant families of entanglement witnesses}


A powerful method of detecting entanglement in the usual scenario are 
entanglement witnesses. An entanglement witness $W$ is a hermitian 
operator with non-negative expectation values for all separable states, 
and a negative expectation value for some entangled state 
$\rho$ \cite{guhne2009entanglement}. We say that $W$ detects $\rho$, 
as $\langle W \rangle_\rho < 0$ proves that $\rho$ is entangled. In 
this section, we show how witnesses can be used in the scrambled 
data scenario.

\subsection{Scrambling-invariant witnesses}

Inspired by the probability distributions of the states defined 
in Eq.~(\ref{psit}), we define a scrambling-invariant family of 
entanglement witnesses. In the most general form, with local 
measurements $\sigma_x \otimes \sigma_x$, $\sigma_y \otimes \sigma_y$, 
and $\sigma_z \otimes \sigma_z$, they are given by
\begin{align}\label{witness}
\begin{split}
W = \hspace*{3pt} 
&\mathds{1} + \alpha \ket{x_1x_2}\bra{x_1x_2}\\
&+ \beta \ket{y_1y_2}\bra{y_1y_2} + \gamma \ket{z_1z_2}\bra{z_1z_2},
\end{split}
\end{align}
where $\ket{x_j} \in \{\ket{+},\ket{-}\}$, $\ket{y_j} \in \{\ket{\text{y}^+},\ket{\text{y}^-}\}$, and $\ket{z_j} \in \{\ket{0},\ket{1}\}$ for $j=1,2$. 

The key observation is that if for fixed $\alpha$, $\beta$ and $\gamma$ this 
yields an entanglement witness, then also every other choice of $x_j$, $y_j$ 
and $z_j$ results in an entanglement witness. This is because using only local 
unitary transformations and the partial transposition, the witnesses can be 
transformed into each other. Consider for example 
$W = \mathds{1} + \alpha \ket{+-}\bra{+-} + \beta \ket{\text{y}^+\text{y}^+}\bra{\text{y}^+\text{y}^+} + \gamma \ket{10}\bra{10}$, 
and the transformations $U_A = \sigma_x$, and $U_B = \mathds{1}$. Then, 
$U_A^\dagger \otimes U_B^\dagger W^{T_A} U_A \otimes U_B 
= \mathds{1} + \alpha \ket{+-}\bra{+-} + \beta \ket{\text{y}^+\text{y}^+}\bra{\text{y}^+\text{y}^+} + \gamma \ket{00}\bra{00}$ 
and one can directly check that any other witness can also be reached.

Indeed, such mappings correspond to permutations of the probabilities, 
as
\begin{align}
\langle W \rangle = 1 + \alpha p_{x_1x_2} + \beta p_{y_1y_2} + \gamma p_{z_1z_2}.
\end{align}
So, for evaluating such a witness from scrambled data, one can just choose the
probabilities appropriately in order to minimize the mean value of the witness.

As a remark, for $\alpha, \beta, \gamma < 0$, the witnesses are additionally 
related to an entropic uncertainty relation for min-entropy 
$S^\infty(\vec{p}) = - \log \max_j p_j$ since the smallest expectation 
value of the corresponding family of entanglement witnesses can be 
written as
\begin{align}
\langle W \rangle = 1 + \alpha e^{-S^{(\infty)}_{xx}(\rho)} + \beta e^{-S^{(\infty)}_{yy}(\rho)} + \gamma e^{-S^{(\infty)}_{zz}(\rho)}.
\end{align}

\subsection{Optimized witnesses}

We want to find optimized $\alpha$, $\beta$ and $\gamma$ such 
that $W$ is an entanglement witness tangent to the space of 
separable states, i.e., there exists a separable state with 
$\langle W \rangle = 0$. In the following analysis, we restrict 
ourselves to only two measurements and witnesses of the form
\begin{align}
W = \mathds{1} + \alpha \ket{++}\bra{++} + \gamma \ket{00}\bra{00}.
\end{align}

First of all, we need to ensure that $\langle W \rangle \geqslant 0$ 
for all separable states. In order to obtain an optimal witness, 
we further need to adjust $\alpha$ and $\gamma$ such that for some 
separable state $\langle W \rangle = 0$.

The optimal values for $\alpha$ and $\gamma$ are found by optimizing
\begin{align}
\min_{\rho_s} \trace (\rho_s W)
\end{align}
for all $\alpha$ and $\gamma$. Because of linearity, we only need to 
consider general pure product states $\ket{\psi_A} \otimes \ket{\psi_B}$ 
where
\begin{align}
\ket{\psi_{A/B}} = \cos\frac{\theta_{A/B}}{2} \ket{0} + e^{i\phi_{A/B}} \sin\frac{\theta_{A/B}}{2} \ket{1}
\end{align}
with $0 \leqslant \theta \leqslant \pi$, $0 \leqslant \phi < 2\pi$. 
It turns out that for $ \frac{\gamma}{\alpha} \geqslant -3-2\sqrt{2}$, the optimal state is given by $\phi_A = \phi_B = 0$ and $\theta_A = \theta_B$, while $\phi_A = \phi_B = 0$ and $\theta_A - \frac{3\pi}{4} = \frac{3\pi}{4} - \theta_B$ needs to be considered in the case of $ \frac{\gamma}{\alpha} \leqslant -3-2\sqrt{2}$.

\begin{figure}[t]
\includegraphics[width=0.9\linewidth]{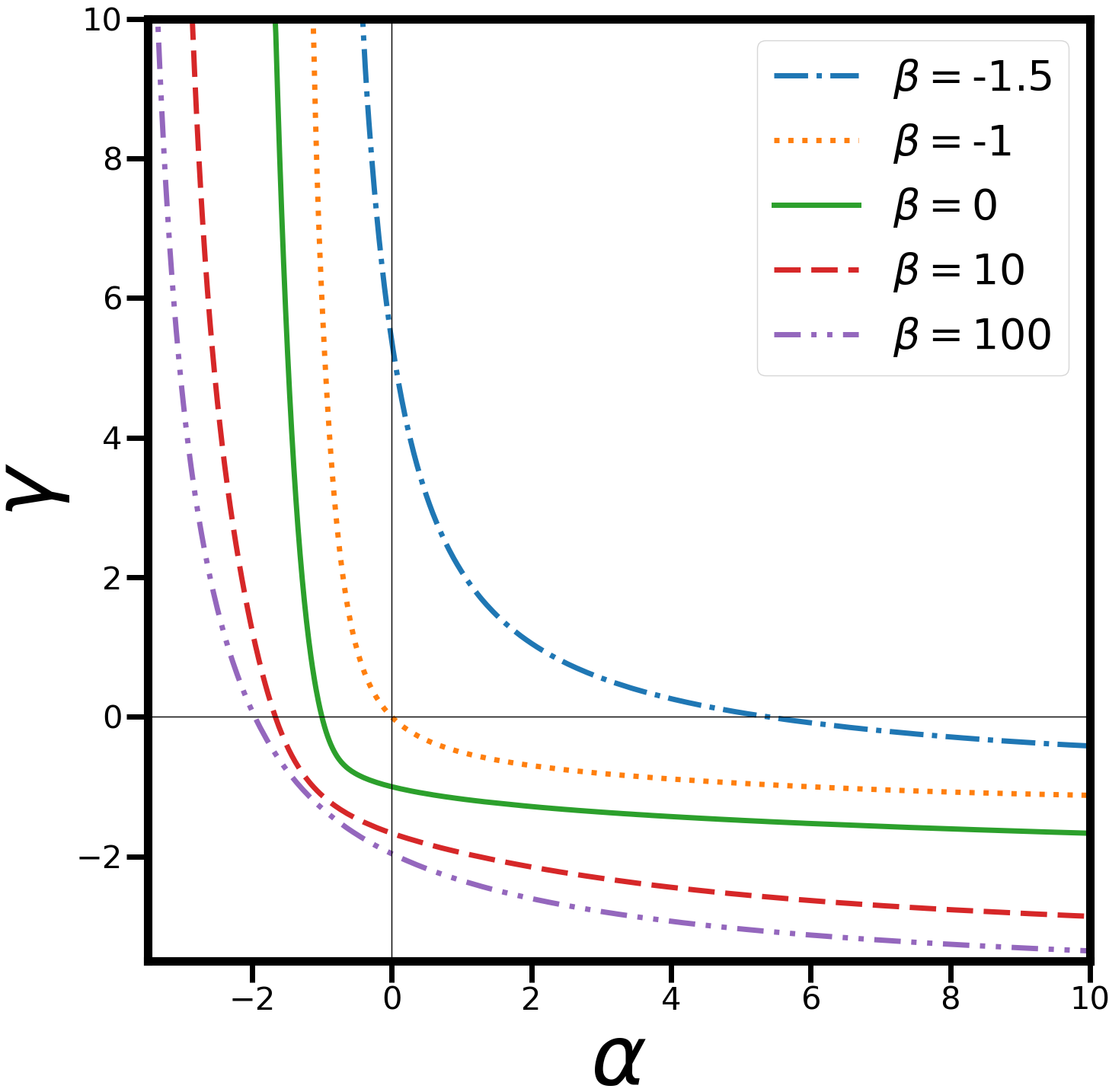}
\caption[justification=raggedright]{Optimized values for the parameters 
$\alpha$ and $\gamma$ for different $\beta$ in entanglement witnesses of 
the form $W = \mathds{1} + \alpha \ket{x_1x_2}\bra{x_1x_2} + \beta \ket{y_1y_2}\bra{y_1y_2} + \gamma \ket{z_1z_2}\bra{z_1z_2}$. Here, optimized means that for 
some separable state $\langle W \rangle = 0$.}\label{abc}
\end{figure}

Finally, we have to ensure that there exist entangled states with $\langle W \rangle < 0$. Since $\langle W \rangle = 1 + \alpha p_{++} + \gamma p_{00}$ and the probabilities are non-negative, either $\alpha$ or $\gamma$ must necessarily be negative. In that case, the eigenvector of $W$ corresponding to the smallest eigenvalue is indeed given by the entangled state
$\ket{\psi_t} = \frac{1}{\sqrt{3+t^2}} (t, 1, 1, 1)^T$ with 
$t=-{(\alpha - 2 \gamma + 2 \sqrt{\alpha^2 - \alpha \gamma + \gamma^2})}/{\alpha}$.

The resulting curve of optimal $\alpha$ and $\gamma$ in the case of $\beta = 0$ 
can thus be obtained analytically and is shown in Fig.~\ref{abc}. More generally, 
for witnesses of the form in Eq.~(\ref{witness}) where $\beta \neq 0$, we find 
the optimal parameters numerically.

\section{Non-convex structure of the non-detectable state space}

For many methods of entanglement detection, it is crucial that the set 
of separable states is convex. For instance, the existence of a witness
for any entangled state $\rho$ relies on this fact. This convexity is also
present in the case of restricted measurements, which are not tomographically
complete. If there is a way to detect the entanglement from a restricted set of
measurements, it can be done with an entanglement witness \cite{curty2005detecting}. 
In this section, we show that this is not the case when only scrambled data is 
available.

In order to test whether there would in principle be a method to detect the 
entanglement of a specific state using only scrambled data from local 
measurements $\sigma_x \otimes \sigma_x$ and $\sigma_z \otimes \sigma_z$, 
we use the fact that the PPT criterion is necessary and sufficient in the 
two-qubit case \cite{horodecki1996ppt}. Thus, we can formulate the problem 
as a family of semi-definite programs (SDPs) \cite{boyd}. We consider the 
problem
\begin{equation}
\begin{aligned}
\min_\rho \: &0 \\
\text{s.t.} &\trace \rho = 1, \\
&\rho \geqslant 0, \\
&\rho^{T_B} \geqslant 0, \\
&\rho \text{ realizes one of the }(4!)^2 \text{ permutations} \\
&\text{\,\,\,\,\, of the given probability distribution} \\
&\text{\,\,\,\,\, for measurements } \sigma_x \otimes \sigma_x \text{ and } \sigma_z \otimes \sigma_z.
\end{aligned}
\end{equation}
This is a so called feasibility problem: If a $\rho$ with the desired properties exist, 
the output of the SDP is zero, and $\infty$ otherwise. If this family of SDPs fails for 
all permutations, then there is no separable state that realizes the same scrambled data 
as the original state. Hence, the entanglement of such a state is detected. Otherwise, 
we call the state {possibly separable}. 

In practice, without scrambled data, around $1.2\%$ of all random states according 
to the Hilbert-Schmidt measure can be shown to be entangled using only local 
measurements $\sigma_x \otimes \sigma_x$ and $\sigma_z \otimes \sigma_z$. In 
the case of scrambled data, we tested approximately $130,000,000$ random 
mixed states and found around $3000$ detectable states using the corresponding 
scrambled data. Note that for the implementation it is possible to reduce the 
number of permutations that need to be considered to just $18$, as local relabeling 
of the outcomes or the exchange of qubits can be neglected.

Out of these states, only six can be detected using the scrambling-invariant 
entanglement witnesses and \textit{none} using the entropic uncertainty 
relations where $q = \tilde{q}$. The reason for this poor performance 
is the non-convex structure of the set of non-detectable states, as we discuss
now.

\begin{figure}[t]
\includegraphics[width=0.9\linewidth]{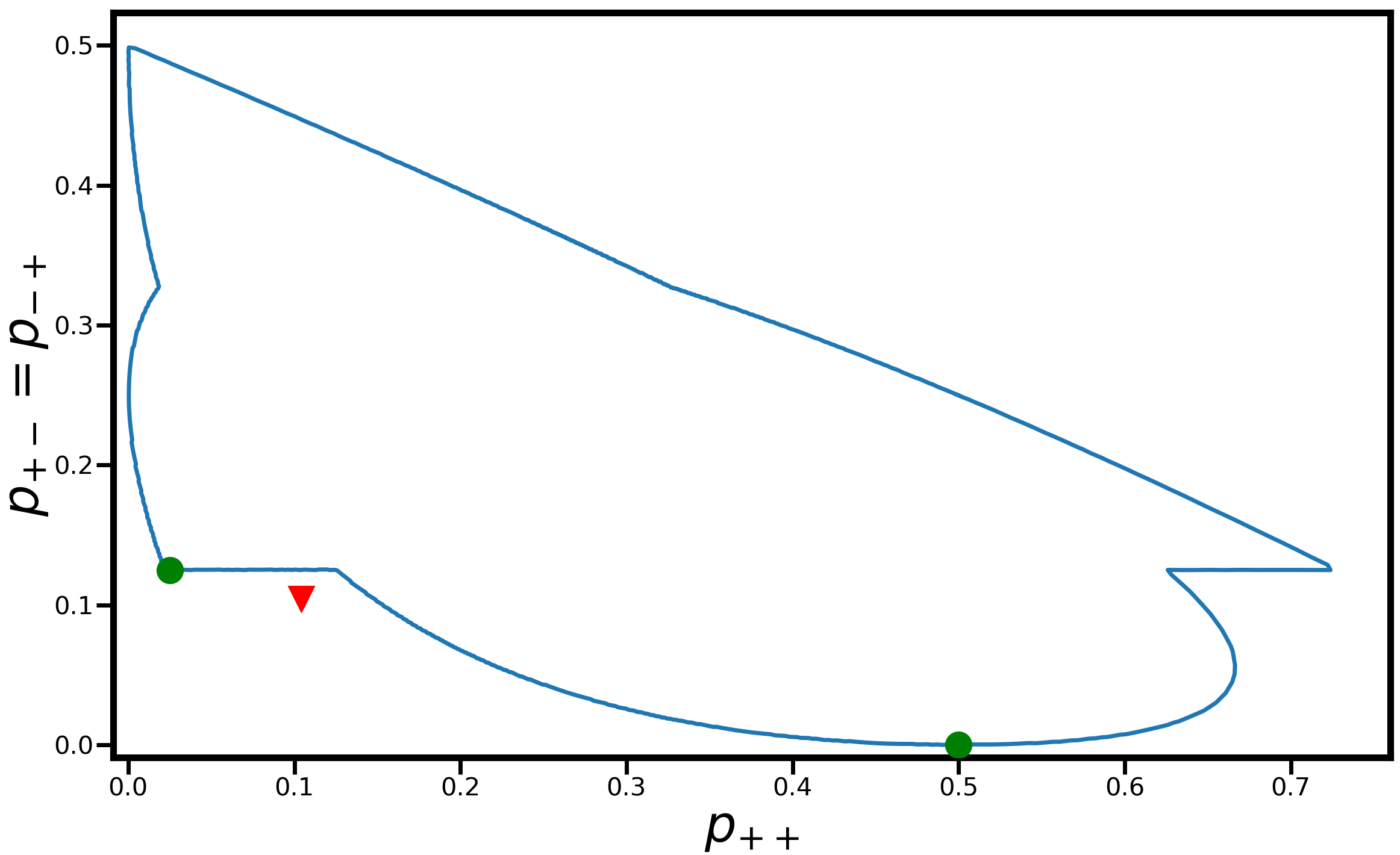}
\caption[justification=raggedright]{Projection of the set of 
possibly separable states (blue line) for local measurements 
$\sigma_x \otimes \sigma_x$ and $\sigma_z \otimes \sigma_z$, 
where $p_{++} = p_{00}$, $p_{+-} = p_{-+} = p_{01} = p_{10}$, 
and $p_{--} = p_{11}$, onto the coordinates $(p_{++},p_{+-})$. 
Clearly, this set is non-convex. The green dots and the red 
triangle correspond to an explicit counterexample to the
convexity, as explained
in the main text.}\label{nonconvex}
\end{figure}

First, we note that the set of possibly separable states is 
star-convex around the maximally mixed state $\frac{\mathds{1}}{4}$.
This can be seen as follows: If a state $\rho$ is part of the set of 
possibly separable states, there is a separable state $\sigma$ that 
realizes the same probability distribution as $\rho$ up to a permutation. 
Then, $\lambda \sigma + (1-\lambda) \frac{\mathds{1}}{4}$ is still 
separable for $0 \leqslant \lambda \leqslant 1$ and realizes the same 
probability distribution as $\lambda \rho + (1-\lambda) \frac{\mathds{1}}{4}$ 
up to the same permutation as before. 

This fact can be used to characterize 
the boundary of the possibly-separable state space by starting with the 
maximally mixed state and mixing it with detectable states until the 
mixture becomes detectable. To illustrate the non-convexity of this set, 
we assume first that it is convex. Then, the intersection with any convex 
set, for example the set of states with probabilities 
$p_{++} = p_{00}$, $p_{+-} = p_{-+} = p_{01} = p_{10}$, 
and $p_{--} = p_{11}$, would again form a convex set. Furthermore, 
the projection onto the coordinates $(p_{++},p_{+-})$ would be convex. This 
projection is shown in Fig.~\ref{nonconvex}. Clearly, it is non-convex, 
and hence, the initial assumption is incorrect. To make this statement 
independent of numerical analysis, we provide an explicit counterexample.

\begin{observation}
The set of possibly separable states for local measurements $\sigma_x \otimes \sigma_x$ and $\sigma_z \otimes \sigma_z$ is non-convex.
\end{observation}

The states $\rho_1 = \frac{1}{4} \left( \mathds{1} \otimes \mathds{1} - \frac{7}{10} (\mathds{1} \otimes \sigma_x + \sigma_x \otimes \mathds{1} \right.$ $\left. + \mathds{1} \otimes \sigma_z + \sigma_z \otimes \mathds{1}) + \frac{1}{2} (\sigma_x \otimes \sigma_x + \sigma_z \otimes \sigma_z + \sigma_x \otimes \sigma_z \right.$ $\left. + \sigma_z \otimes \sigma_x) \right)$ and $\rho_2 = \ket{\Phi^+}\bra{\Phi^+}$ where $\ket{\Phi^+} = \frac{1}{\sqrt{2}} (\ket{00} + \ket{11})$ realize probability distributions corresponding to the left and right green dot in Fig.~\ref{nonconvex}, respectively. While $\rho_1$ is separable, the product state $\ket{+}\ket{0}$ realizes the same scrambled data as $\rho_2$ and hence, $\rho_2$ is possibly separable. Thus, they are part of the possibly separable state space. However, the mixture $\rho = \frac{5}{6} \rho_1 + \frac{1}{6} \rho_2$, shown as a red triangle in Fig.~\ref{nonconvex}, is detectable. The scrambled data of the corresponding probability distribution $p_{++} = p_{+-}
 = p_{-+} = p_{00} = p_{01} = p_{10} = \frac{5}{48}$ and $p_{--} = p_{11} = \frac{33}{48}$ cannot origin from a separable state. The witnesses $W = \mathds{1} \pm \sigma_x \otimes \sigma_x \pm \sigma_z \otimes \sigma_z$ certify the entanglement for all permutations.\\

For general two-qubit states, the same procedure can be applied. Indeed, 
mixing the $3000$ random detectable states with white noise such that 
they are barely compatible with scrambled data from separable states, 
can be used to characterize the boundary of the set of possibly separable 
states. Mixing pairs of these states with equal weights leads in some 
cases to detectable states, also witnessing the non-convex structure.


\section{Conclusion}


We have introduced the concept of scrambled data, meaning that the assignment 
of probabilities to outcomes of the measurements is lost. Clearly, this restriction 
limits the possibilities of entanglement detection. Nevertheless, we have shown
that using entropies and entanglement witnesses one can still detect the entanglement
in some cases. These methods are limited, however, as the set of states whose 
scrambled data can be realized by separable states is generally not convex.

There are several directions in which our work may be extended or generalized. 
First, one may consider more general scenarios than the two-qubit situation 
considered here, such as the case of three or more particles. Second, it would
be interesting to study our results on entropies further, in order to derive
systematically entropic uncertainty relations for various entropies. Such entropic
uncertainty relations find natural applications in the security analysis of quantum
key distribution and quantum information theory. Finally, it would be intriguing
to connect our scenario to Bell inequalities.  This could help to relax assumptions 
on the data for non-locality detection and device-independent quantum information
processing.

\section{Acknowledgments}

We thank Xiao-Dong Yu, Ana Cristina Sprotte Costa and Roope Uola for fruitful discussions.
This work was supported by the DFG, the ERC (Consolidator 
Grant No. 683107/TempoQ), the Asian Office of 
Aerospace (R\&D grant FA2386-18-1-4033) and the House of Young Talents Siegen.

\end{document}